\documentclass[a4paper,UKenglish,cleveref, autoref]{lipics-v2019}

\usepackage[utf8]{inputenc}
\usepackage{amsmath}
\usepackage{amsfonts}
\usepackage{amssymb}
\usepackage{todonotes}

\newcommand{\Oh}{\mathcal{O}}
\newcommand{\NP}{\mbox{NP}}
\newcommand{\PP}{\mbox{P}}
\newcommand{\xig}{\xi_{\mbox{\tiny $G$}}}

\title{Computing the largest bond of a graph}

\titlerunning{Computing the largest bond of a graph}

\author{Gabriel L. Duarte}{Fluminense Federal University, RJ, Brazil}{gabrield@id.uff.br}{}{}

\author{Daniel Lokshtanov}{University of California Santa Barbara, CA, USA}{daniello@ucsb.edu}{}{}

\author{Lehilton L. C. Pedrosa}{University of Campinas, SP, Brazil}{lehilton@ic.unicamp.br}{https://orcid.org/0000-0003-1001-082X}{}

\author{Rafael C. S. Schouery}{University of Campinas, SP, Brazil}{rafael@ic.unicamp.br}{https://orcid.org/0000-0002-0472-4810}{}

\author{Uéverton S. Souza\footnote{corresponding author}}{Fluminense Federal University, RJ, Brazil}{ueverton@ic.uff.br}{https://orcid.org/0000-0002-5320-9209}{}

\authorrunning{G. L. Duarte, D. Lokshtanov, L. L. C. Pedrosa, R. C. S. Schouery and U. S. Souza}

\Copyright{Gabriel L. Duarte, Daniel Lokshtanov, Lehilton L. C. Pedrosa, Rafael C. S. Schouery and Uéverton S. Souza}



\category{}

\relatedversion{}

\supplement{}

\funding{Supported by Grant 2015/11937-9, São Paulo Research Foundation (FAPESP) and by Grant E-26/203.272/2017, Rio de Janeiro Research Foundation (FAPERJ) and by Grant 308689/2017-8, 425340/2016-3, 313026/2017-3, 422829/2018-8, 303726/2017-2, National Council for Scientific and Technological Development (CNPq).}

\acknowledgements{We thank the organizers of WoPOCA 2017 for the opportunity to bring together some of the co-authors of this paper.}

\nolinenumbers 

\hideLIPIcs  

\begin{document}

\maketitle

\begin{abstract}
A bond of a graph $G$ is an inclusion-wise minimal disconnecting set of $G$, i.e., bonds are cut-sets that determine cuts $[S,V\setminus S]$ of $G$ such that $G[S]$ and $G[V\setminus S]$ are both connected. Given $s,t\in V(G)$, an $st$-bond of $G$ is a bond whose removal disconnects $s$ and $t$.
Contrasting with the large number of studies related to maximum cuts, there are very few results regarding the largest bond of general graphs.
In this paper, we aim to reduce this gap on the complexity of computing the largest bond and the largest $st$-bond of a graph.
Although cuts and bonds are similar, we remark that computing the largest bond of a graph tends to be harder than computing its maximum cut.
We show that {\sc Largest Bond} remains \NP-hard even for planar bipartite graphs, and it does not admit a constant-factor approximation algorithm, unless $\PP = \NP$.
We also show that {\sc Largest Bond} and {\sc Largest $st$-Bond} on graphs of clique-width $w$ cannot be solved in time $f(w)\times n^{o(w)}$ unless the
Exponential Time Hypothesis
fails, but they can be solved in time $f(w)\times n^{O(w)}$. In addition, we show that both problems are fixed-parameter tractable when parameterized by the size of the solution, but they do not admit polynomial kernels unless NP $\subseteq$ coNP/poly.

\end{abstract}

\section{Introduction}

Let $G=(V,E)$ be a simple, connected, undirected graph.  A \emph{disconnecting set} of $G$ is a set of edges $F\subseteq E(G)$ whose removal disconnects $G$. The edge-connectivity of $G$ is
$\kappa'(G) = \min\{|F|: \mbox{$F$ is a disconnecting set of $G$}\}$.
A cut $[S,T]$ of $G$ is a partition of $V$ into two subsets $S$ and $T=V\setminus S$. The cut-set $\partial(S)$ of a cut $[S,T]$ is the set of edges that have one endpoint in $S$ and the other endpoint in $T$; these edges are said to cross the cut. In a connected graph, each cut-set determines a unique cut.
Note that every cut-set is a disconnecting set, but the converse is not true.
An inclusion-wise minimal disconnecting set of a graph is called a \emph{bond}. It is easy to see that every bond is a cut-set, but there are cut-sets that are not bonds. More precisely, a nonempty set of edges $F$ of $G$ is a bond if and only if $F$ determines a cut $[S,T]$ of $G$ such that $G[S]$ and $G[T]$ are both connected. Let~$s,t\in V(G)$. An $st$-bond of $G$ is a bond whose removal disconnects $s$ and $t$.

In this paper, we are interested in the complexity aspects of the following problem.

\medskip
\noindent    \fbox{
        \parbox{.95\textwidth}{
\noindent
\textsc{Largest Bond}\\
\noindent
\textbf{Instance}: A graph $G=(V,E)$; a positive integer $k$.\\
\noindent
\textbf{Question}: Is there a proper subset $S\subset V(G)$ such that  $G[S]$ and $G[V\setminus S]$ are connected and $|\partial(S)|\geq k$?
}
}
\medskip

We also consider \textsc{Largest $st$-Bond}, where given a graph $G=(V,E)$, vertices $s,t\in V(G)$, and a positive integer $k$, we are asked whether $G$ has an $st$-bond of size at least $k$.

\noindent    \fbox{
        \parbox{.965\textwidth}{
\noindent
{\sc \textsc{Largest $st$-Bond}}

\noindent
\textbf{Instance}: A graph $G=(V,E)$; vertices $s,t\in V(G)$; a positive integer $k$.


\noindent
\textbf{Question}: Is there a proper subset $S\subset V(G)$ with $s\in S$ and $t\notin S$, such that  $G[S]$ and $G[V\setminus S]$ are connected and $|\partial(S)|\geq k$?
}
}

A minimum (maximum) cut of a graph $G$ is a cut with cut-set of minimum (maximum) size. Every minimum cut is a bond, thus a minimum bond is also a minimum cut of $G$, and it can be found in polynomial time using the classical Edmonds–Karp algorithm~\cite{EdmondsKarpAlgorithm}. Besides that, minimum $st$-bonds are well-known structures, since they are precisely the $st$-cuts involved in the Gomory-Hu trees~\cite{GomoryHu}.

Regarding bonds on planar graphs, a folklore theorem states that if $G$ is a connected planar graph, then a set of edges is a cycle in $G$ if and only if it corresponds to a bond in the dual graph of~$G$~\cite{LivroPlanarBond}. Note that each cycle separates the faces of $G$ into the faces in the interior of the cycle and the faces of the exterior of the cycle, and the duals of the cycle edges are exactly the edges that cross from the interior to the exterior~\cite{oxley2006matroid}. Consequently, the girth of a planar graph equals the edge connectivity of its dual~\cite{CHO20072456}.

Although cuts and bonds are similar, computing the largest bond of a graph seems to be harder than computing its maximum cut. {\sc Maximum Cut} is \NP-hard in general~\cite{garey2002computers}, but becomes polynomial for planar graphs~\cite{Hadlock}. On the other hand, finding a longest cycle in a planar graph is \NP-hard, implying that finding a largest bond of a planar multigraph (or of a simple edge-weighted planar graph) is \NP-hard. In addition, it is well-known that if a simple planar graph is $3$-vertex-connected, then its dual is a simple planar graph. In~1976, Garey, Johnson, and Tarjan~\cite{PlanarHamiltonian} proved that the problem of establishing whether a $3$-vertex-connected planar graph is Hamiltonian is \NP-complete, thus, as also noted by Haglin and Venkatesan~\cite{Haglin},  finding the largest bond of a simple planar graph is also \NP-hard, contrasting with the polynomial-time solvability of {\sc Maximum Cut} on planar graphs. 

From the point of view of parameterized complexity, it is well known that {\sc Maximum Cut} can be solved in FPT time when parametrized by the size of the solution~\cite{mahajan1999parameterizing}, and since every graph has a cut with at least half the edges~\cite{erdos1965some}, it follows that it has a linear kernel. Concerning approximation algorithms, a $1/2$-approximation algorithm can be obtained by randomly partitioning the set vertices into two parts, which induces a cut-set whose expected size is at least half of the number of edges~\cite{mitzenmacher_upfal_2005}. The best-known result is the seminal work of Goemans and Williamson~\cite{goemans1995semidef}, who gave a $0.878$-approximation based on semidefinite programming. This has the best approximation factor unless the Unique Games Conjecture fails~\cite{khot2007hard}.
To the best of our knowledge, there is no algorithmic study regarding the parameterized complexity of computing the largest bond of a graph as well as the approximability of the problem.

A closely related problem is the {\sc Connected Max Cut}~\cite{HajiaghayiKMPS15}, which asks for a cut $[S, T]$ of a given a graph $G$ such that $G[S]$ is connected, and that the cut-set $\partial(S)$ has size at least~$k$.
Observe that a bond induces a feasible solution of {\sc Connected Max Cut}, but not the other way around, since $G[T]$ may be disconnected. Indeed, the size of a largest bond can be arbitrarily smaller than the size of the maximum connected cut; take, e.g., a star with $n$ leaves.
For {\sc Connected Max Cut} on general graphs, there exists a $\Omega(1/\log n)$-approximation~\cite{GandhiHKPS18}, where $n$ is the number of vertices. Also, there is a constant-factor approximation with factor~$1/2$ for graphs of bounded treewidth~\cite{ShenLN18}, and a polynomial-time approximation scheme for graphs of bounded genus~\cite{HajiaghayiKMPS15}.

Recently, Saurabh and Zehavi~\cite{SaurabhZ19} considered a generalization of {\sc Connected Max Cut}, named {\sc Multi-Node Hub}. In this problem, given numbers $l$ and $k$, the objective is to find a cut $[S, T]$ of $G$ such that $G[S]$ is connected, $|S| = l$ and $|\partial(S)| \ge k$.
They observed that the problem is $W[1]$-hard when parameterized on~$l$, and gave the first parameterized algorithm for the problem with respect to the parameter~$k$. We remark that the $W[1]$-hardness also holds for {\sc Largest Bond} parameterized by $|S|$.

Since every nonempty bond determines a cut $[S,T]$ such that $G[S]$ and $G[T]$ are both connected, every bond of $G$ has size at most $|E(G)|-|V(G)|+2$. A graph~$G$ has a bond of size $|E(G)|-|V(G)|+2$ if and only if $V(G)$ can be partitioned into two parts such that each part induces a tree. Such graphs are known as \emph{Yutsis graphs}. The set of planar Yutsis graphs is exactly the dual class of  Hamiltonian planar graphs.
According to Aldred, Van Dyck, Brinkmann, Fack, and McKay~\cite{aldred2009graph}, cubic Yutsis graphs appear in the quantum theory of angular momenta as a graphical representation of general recoupling coefficients. They can be manipulated following certain rules in order to generate the so-called summation formulae for the general recoupling coefficient (see \cite{biedenharn1981racah,VANDYCK20071506,yutsis1962mathematical}).

There are very few results about the largest bond size in general graphs. In 2008, Aldred, Van Dyck, Brinkmann, Fack, and McKay~\cite{aldred2009graph} showed that if a Yutsis graph is regular with degree $3$, the partition of the vertex set from the largest bond will result in two sets of equal size. In 2015, Ding, Dziobiak and Wu~\cite{ding2016large} proved that any simple $3$-connected graph $G$ will have a largest bond with size at least $\frac{2}{17}\sqrt{\log n}$, where $n=|V(G)|$. In 2017, Flynn~\cite{flynn2017largest} verified the conjecture that any simple $3$-connected graph $G$ has a largest bond with size at least $\Omega(n^{\log_32})$ for a variety of graph classes including planar graphs.

In this paper, we complement the state of the art on the problem of computing the largest bond of a graph.
Preliminarily, we observe that while {\sc Maximum Cut} is trivial for bipartite graphs, {\sc Largest Bond} remains \NP-hard for such a class of graphs, and we also present a general reduction that allows us to observe that {\sc Largest Bond} is \NP-hard for several classes for which {\sc Maximum Cut} is \NP-hard. Using this framework, we are able to show that {\sc Largest Bond} on graphs of clique-width $w$ cannot be solved in time $f(w)\times n^{o(w)}$ unless the ETH fails.
Moreover, we show that {\sc Largest Bond} does not admit a constant-factor approximation algorithm, unless $\PP = \NP$, and thus is asymptotically harder to approximate than {\sc Maximum Cut}.

As for positive results, the main contributions of this work concern the parameterized complexity of {\sc Largest Bond}. Using win/win approaches, we consider the strategy of preprocessing the input in order to bound the treewidth of the resulting instance. After that, by presenting a dynamic programming algorithm for {\sc Largest Bond} parameterized by the treewidth, we show that the problem is fixed-parameter tractable when parameterized by the size of the solution. Finally, we remark that  {\sc Largest Bond} and  {\sc Largest $st$-Bond} do not admit polynomial kernels, unless NP $\subseteq$ coNP/poly.
%

\section{Intractability results}

In this section, we discuss aspects of the hardness of computing the largest bond.
Notice that {\sc Largest Bond} is  Turing reducible to {\sc Largest $st$-Bond}. Therefore, the results presented in this section also holds for {\sc Largest $st$-Bond}.

Although {\sc Maximum Cut} is trivial for bipartite graphs, we first observe that the same does not apply to compute the largest bond.
%


\begin{theorem}\label{bipartite}
{\sc Largest Bond} is \NP-complete for planar bipartite graphs.
\end{theorem}

\begin{proof}
It is well-known that a connected planar graph is Eulerian if and only if its dual graph is bipartite~\cite{WELSH1969375}. In 1994, Picouleau~\cite{PICOULEAU1994463} proved that deciding whether a $4$-regular planar graph has a Hamiltonian cycle is \NP-complete. Thus, to determine the size of the largest bond of a planar bipartite multigraph is \NP-complete. In order to obtain a simple planar bipartite graph, it is enough to subdivide each edge of the graph; notice that this operation preserves the size of the largest bond of the graph. Therefore, determining the size of the largest bond of a simple planar bipartite graph is \NP-complete.
\end{proof}

\begin{theorem}\label{w1hard}
Let $G$ be a simple bipartite graph and $\ell\in \mathbb{N}$. To determine the largest bond $\partial(S)$ of $G$ with $|S|=\ell$ is $W[1]$-hard with respect to $\ell$.
\end{theorem}

\begin{proof}
From an instance $H$ of {\sc $k$-Independent Set} on regular graphs we first construct a multigraph $G'$ by adding an edge between any pair of vertices.
Finally, we obtain a simple graph $G$ by subdividing every edge of $G'$.
Notice that $H$ has an independent set of size $k$ if and only if $G$ has a bond $\partial(S)$ of size $dk+k(n-k)$ with $|S|=k + {\genfrac(){0pt}{2}{k}{2}}$, where $d$ is the vertex degree of $H$.
\end{proof}

Next, we present a general framework for reducibility from {\sc Maximum Cut} to {\sc Largest Bond}, by defining a special graph operator $\psi$ such that {\sc Maximum Cut} on a graph class~$\mathcal{F}$ is reducible to {\sc Largest Bond} on the image of $\mathcal{F}$ via~$\psi$. An interesting particular case occurs when $\mathcal{F}$ is closed under $\psi$ (for instance, chordal graphs are closed under $\psi$).

\begin{definition}
Let $G$ be a graph and let $n=V(G)$. The graph $\psi(G)$ is constructed as follows: 
\begin{enumerate}
\item[(i)] create $n$ disjoint copies $G_1, G_2, \ldots, G_n$ of $G$;
\item[(ii)] add vertices $v_a$ and $v_b$;
\item[(iii)] add an edge between $v_a$ and $v_b$;
\item[(iv)] add all  possible edges between $V(G_1 \cup G_2 \cup \ldots \cup G_n)$ and~$\{v_a,v_b\}$.
\end{enumerate}
\end{definition}

\begin{definition}
A set of graphs~$\mathcal{G}$ is closed under operator $\psi$ if whenever $G \in \mathcal{G}$, then $\psi(G)\in \mathcal{G}$.
\end{definition}


\begin{theorem}\label{framework}
{\sc Largest Bond} is \NP-complete for any graph class $\mathcal{G}$ such that:
\begin{itemize}
\item $\mathcal{G}$ is closed under operator $\psi$; and 
\item {\sc MaxCut} is \NP-complete for graphs in $\mathcal{G}$.
\end{itemize}
\end{theorem}

\begin{proof}
Let $G\in \mathcal{G}$, $n = |V(G)|$, and $H=\psi(G)$. By~(i), $H\in\mathcal{G}$. Suppose $G$ has a cut $[S,V(G)\setminus S]$ of size $k$, and let $S_1$, $S_2$, \ldots, $S_n$ be the copies of $S$ in $G_1, G_2, \ldots, G_n$, respectively. If $S'=\{v_a\}\cup S_1\cup S_2 \cup \ldots \cup S_n$, then $[S',V(H)\setminus S']$ defines a bond $\partial(S')$ of~$H$ of size at least $nk+n^2+1$.
Conversely, suppose $H$ has a bond $\partial(S')$ of size at least $nk+n^2+1$. We consider the following cases: $(a)$~If $\{v_a, v_b\}\subseteq S'$, then for all copies $G_i$ but one we have $V(G_i)\subseteq S'$, as otherwise the graph induced by $V(H)\setminus S'$ would not be connected, and~$\partial(S')$ would not be a bond. Thus, $V(H)\setminus S' \subseteq V(G_j)$ for some $j$, then the size of $\partial(S')$ is smaller than $nk+n^2+1$, a contradiction. $(b)$~If $v_a\in S'$ and $v_b\notin S'$,
then  $\{v_a,v_b\}$ is incident with exactly $n^2+1$ edges crossing $[S',V(H)\setminus S']$, which implies that at least one copy $G_i$ has $k$ or more edges crossing $[S',V(H)\setminus S']$. Therefore, $G$ has a cut of size at least $k$.
\end{proof}

\begin{corollary}\label{subclasses_NPc}
{\sc Largest Bond} is \NP-complete for the following classes:
\begin{enumerate}
  \item chordal graphs;
  \item co-comparability graphs;
  \item $P_5$-free graphs.
\end{enumerate}
\end{corollary}

\begin{proof}
Bodlaender and Jansen~\cite{bodlaender2000complexity} proved that {\sc Maximum Cut} is \NP-complete when restricted to split and co-bipartite graphs. Since split graphs are chordal and co-bipartite graphs are $P_5$-free and co-comparability graphs, the \NP-completeness also holds for these classes.

Now we have to show that the classes are closed under $\psi$.

\emph{(1.)} A graph is chordal if every cycle of length at least $4$ has a chord. Let $G$ be a chordal graph. Notice that the disjoint union of $G_1, G_2, \ldots, G_n$ is also chordal. In addition, no chordless cycle of length at least $4$ may contain either $v_a$ or $v_b$ because both vertices are universal. Therefore,~$\psi(G)$ is chordal.

\emph{(2.)} A graph is a co-comparability if it is the intersection graph of curves from a line to a parallel line. Let $G$ be a co-comparability graph. Notice that the class of co-comparability graphs is closed under disjoint union. Thus, in order to conclude that $\psi(G)$ is co-comparability, it is enough to observe that from a representation of curves (from a line to a parallel line) of the disjoint union of $G_1, G_2, \ldots, G_n$, one can construct a representation of $\psi(G)$ by adding two concurrent lines (representing $v_a$ and $v_b$) crossing all curves.

\emph{(3.)} The disjoint union of $P_5$-free graphs is also $P_5$-free. In addition, no induced $P_5$ contains either $v_a$ or $v_b$ because both vertices are universal. Then, the class of $P_5$-free graphs is closed under~$\psi$.
\end{proof}

\subsection{Algorithmic lower bound for clique-width parameterization}

The {\em clique-width} of a graph~$G$, denoted by~$cw(G)$, is defined as the minimum number of labels needed to construct~$G$, using the following four operations:

\begin{enumerate}
  \item Create a single vertex~$v$ with an integer label~$\ell$ (denoted by~$\ell(v)$);

  \item Take the disjoint union (i.e., co-join) of two graphs (denoted by~$\oplus$);

  \item Join by an edge every vertex labeled~$i$ to every vertex labeled~$j$ for~$i \neq j$ (denoted by~$\eta(i,j)$);

  \item Relabel all vertices with label~$i$ by label~$j$ (denoted by~$\rho(i,j)$).
\end{enumerate}

An algebraic term that represents such a construction of $G$ and uses at most~$w$ labels is said to be a {\em $w$-expression} of $G$, and the clique-width of $G$ is the minimum~$w$ for which $G$ has a $w$-expression.

Graph classes with bounded clique-width include cographs, distance-hereditary graphs, graphs of bounded treewidth, graphs of bounded branchwidth, and graphs of bounded rank-width.

In the '90s,  Courcelle, Makowsky, and Rotics~\cite{courcelle2000linear} proved that all problems expressible in MS1-logic are fixed-parameter tractable when parameterized by the clique-width of a graph and the logical expression size.
The applicability of this meta-theorem has made clique-width become one of the most studied parameters in parameterized complexity. However, although several problems are MS1-expressible, this is not the case with {\sc Maximum Cut}.

In 2014, Fomin, Golovach, Lokshtanov and Saurabh~\cite{fomin2014almost} showed that {\sc Maximum Cut} on a graph of clique-width $w$ cannot be solved in time $f(w)\times n^{o(w)}$ for any function $f$ of $w$ unless Exponential Time Hypothesis (ETH) fails. Using operator $\psi$, we are able to extend this result to {\sc Largest Bond}.

\begin{lemma}\label{lb_cw}
{\sc Largest Bond} on graphs of clique-width $w$ cannot be solved in time $f(w)\times n^{o(w)}$ unless the ETH fails.
\end{lemma}

\begin{proof}
{\sc Maximum Cut} cannot be solved in time $f(w)\times n^{o(w)}$ on graphs of clique-width $w$, unless Exponential Time Hypothesis (ETH) fails~\cite{fomin2014almost}. Therefore, by the polynomial-time reduction presented in Theorem~\ref{framework}, it is enough to show that the clique-width of $\psi(G)$ is upper bounded by a linear function of the clique-width of $G$.

If $G$ has clique-width $w$, then the disjoint union $H_1 = G_1\oplus G_2 \oplus \ldots \oplus G_n$ has clique-width~$w$.
Suppose that all vertices in $H_1$ have label $1$.
Now, let $H_2$ be the graph isomorphic to a $K_2$ such that $V(H)=\{v_a,v_b\}$, and $v_a,v_b$ are labeled with $2$.
In order to construct~$\psi(G)$ from $H_1\oplus H_2$ it is enough to apply the join $\eta(1,2)$. Thus, $\psi(G)$ has clique-width equals~$w$.
\end{proof}

\subsection{Inapproximability}

While the maximum cut of a graph has at least a constant fraction of the edges,
the size of the largest bond can be arbitrarily smaller than the number of
edges; take, e.g., a cycle on $n$ edges, for which a largest bond has size $2$.
This discrepancy is also reflected on the approximability of the problems.
Indeed, we show that {\sc Largest Bond} is strictly harder to approximate than
{\sc Maximum Cut}.
To simplify the presentation, we consider a weighted version of the problem in
which edges are allowed to have weights $0$ or $1$; the hardness results will
follow for the unweighted case as well. In the {\sc Binary Weighted Largest Bond}, the
input is given by a connected weighted graph $H$ with weights $w : E(H) \to
\{0,1\}$. The objective is to find a bond whose total weight is maximum.

Let $G$ be a graph on $n$ vertices and whose maximum cut has size $k$. Next, we
define the $G$-\emph{edge embedding} operator~$\xig$. Given a connected weighted
graph $H$, the weighted graph~$\xig(H)$ is constructed by replacing each edge
$\{u,v\} \in E(H)$ with weight~$1$ by a copy of~$G$, denoted by $G_{uv}$, whose
edges have weight $1$, and, for each vertex~$t$ of $G_{uv}$, new edges~$\{u,t\}$
and~$\{v,t\}$, both with weight $0$.

We can also apply the $G$-edge embedding operation on the graph $\xig(H)$, then
on $\xig(\xig(H))$, and so on. In what follows, for an integer $h \ge 0$, denote
by $\xig^h(H)$ the graph resulting from the operation that receives a graph $H$
and applies $\xig$ successively $h$ times. Notice that $\xig^h(H)$ can be constructed in $\Oh(|V(G)|^{h+1})$ time.
For some $j$, ${0 \le j \le h-1}$, observe that an edge ${\{u,v\} \in
E(\xig^j(H))}$ will be replaced by a series of vertices added in iterations $j +
1, j+ 2, \dots, h$. These vertices will be called the \emph{descendants}
of~$\{u,v\}$, and will be denoted by~$V_{uv}$.

Let $K_2$ be the graph composed of a single edge $\{u,v\}$, and consider the
problem of finding a bond of $\xig(K_2)$ with maximum weight. Since edges
connecting $u$ or $v$ have weight $0$, one can assume that $u$ and $v$ are in
different sides of the bond, and the problem reduces to finding a  maximum cut
of~$G$.
In other words, the operator $\xig$ embeds an instance $G$ of {\sc Maximum Cut}
into an edge~$\{u,v\}$ of~$K_2$.

This suggests the following strategy to solve an instance of {\sc Maximum Cut}.
For some constant integer $h \ge 1$, calculate $H = \xig^h(K_2)$, and obtain a
bond~$F$~of~$H$ with maximum weight.
Note that, to solve $H$, one must solve embedded instances of {\sc Maximum Cut}
in multiple levels simultaneously. For a level~$j$, ${1 \le j \le h - 1}$, each
edge $\{u,v\} \in E(\xig^j(K_2))$ with weight~$1$ will be replaced by a graph
$G_{uv}$ which is isomorphic to $G$.
In Lemma~\ref{lem-edge-cut} below, we argue that $F$ is such that either
$V(G_{uv}) \cup \{u,v\}$ are all in the same side of the cut, or~$u$~and~$v$ are in
distinct sides. In the latter case, the edges of $F$ that separate $u$ and $v$
will induce a cut of $G$.

In the remaining of this section, we consider a constant integer $h \ge 0$.
Then, we define $H^j = \xig^j(K_2)$ for every $j$, ${0 \le j \le h}$, and $H =
H^h$.
Also, we write $[S,T]$ to denote the cut induced by a bond~$F$~of~$H$.

\begin{definition}
Let $F$ be a bond~of~$H$ with cut $[S,T]$.
We say that an edge ${\{u,v\} \in E(H^j)}$ with weight~$1$ is \emph{nice for~$F$}
if either
\begin{itemize}
  \item $|\{u, v\} \cap S| = 1$, or
  \item $(\{u, v\} \cup V_{uv}) \subseteq S$, or
  \item $(\{u, v\} \cup V_{uv}) \subseteq T$.
\end{itemize}
Also, we say that $F$ is \emph{nice} if, for every $j$, ${0 \le j \le
h-1}$, and every edge $\{u,v\} \in E(H^j)$ with weight~$1$, $\{u,v\}$ is nice
for~$F$.
\end{definition}

\begin{lemma}\label{lem-edge-cut}
There is a polynomial-time algorithm that receives a bond $F$, and finds a nice
bond ${F'}$ such that $w(F') = w(F)$.
\end{lemma}

\begin{proof}
Let $[S,T]$ be the cut induced by $F$ and let $j^*$ be minimum such that there
exists an edge $\{u,v\} \in H^{j^*}$ with weight~$1$ which is not nice for $F$.
Then ${|\{u, v\} \cap S| \ne 1}$. Assume, without loss of generality that $u, v
\in S$. In this case,  $U := V_{uv} \cap T$ is not empty. Since removing
vertices $\{u, v\}$ disconnects $U$, and $T$ must be connected, it follows that
${U = T}$.
This implies that ${N(T) \subseteq (V_{uv} \setminus T)\cup \{u, v\}}$.

We will construct a bond $F'$ of $H$ with cut $[S', T']$.
Let $S'$ be the set of vertices in the connected component of $H[S \setminus
\{v\}]$ which contains~$u$, and ${T' = V(H) \setminus S'}$.
Since $H[S]$ is connected, so must be $H[S \setminus S']$. Also, each vertex
of~$U$ is adjacent~to~$v$, thus $H[(S \setminus S') \cup U]$ is connected.
Observe that $T' = (S \setminus S') \cup U$, so indeed the cut~$[S', T']$
induces a bond ${F' = \partial(S')}$.
Observe that any edge that appears only in $F$ or only in $F'$ is adjacent
to~$v$. Since such edges have weight~$0$, this implies~$w(F) = w(F')$.

To complete the proof, we claim that if for some $j$, $0 \le j \le h$, there
exists an edge $\{u,v\} \in H^{j}$ with weight~$1$ which is not nice for~$F'$,
then $j > j^*$. If this claim holds, then we need to repeat the previous
procedure at most $h$ times before obtaining a nice bond~$F'$.

To prove the claim, consider an edge $\{s,t\} \in H^{j}$ which is not nice
for~$F'$. Suppose, for a contradiction, that ${V_{st} \cap V_{uv} = \emptyset}$.
There are two possibilities.
If $s, t \in S'$, then $V_{st} \subseteq S'$;
if $s, t \in T'$, then $V_{st} \subseteq S \setminus S' \subseteq T'$.
In either possibility, $\{s,t\}$ is nice for~$F'$. This is a contradiction, and
thus ${V_{uv} \cap V_{st} \ne \emptyset}$.

The statement ${V_{uv} \cap V_{st} \ne \emptyset}$ can only happen if $V_{uv}
\subseteq V_{st}$ or $V_{st} \subseteq V_{uv}$. If $V_{uv} \subseteq V_{st}$,
then $U \subseteq V_{st}$ and $s, t \in S$. This implies that $\{s,t\}$ is not
nice for~$F$. But in this case $j < j^*$, contradicting the choice of~$j^*$.
Therefore, $V_{st} \subseteq V_{uv}$, and $j > j^*$, proving our claim.
\end{proof}

In the following, assume that $F$ is a nice bond with cut $[S, T]$.
Consider a level $j$, ${0 \le j \le h}$, and an edge $\{u,v\} \in E(H^j)$ with
weight~$1$ such that ${|\{u, v\} \cap S| = 1}$.
If $j < h$, then we define $F_{uv}$ to be the subset of edges in $F$ which are
incident with some vertex of~$V_{uv}$; if $j = h$, then we define $F_{uv} = \{
\{u,v\}\}$.
Note that, because $F$ is nice, if ${|\{u, v\} \cap S| \ne 1}$, then no edge
of~$F$ is incident with $V_{uv}$.

Suppose now that ${|\{u, v\} \cap S| = 1}$ for some edge $\{u,v\} \in E(H^j)$ with
weight~$1$ and $0 \le j \le h - 1$.
In this case, $F$ induces a cut-set of~$G_{uv}$. Namely, define
$
{\hat S_{uv} := S \cap V(G_{uv})}
$
and
$
{\hat T_{uv} := T \cap V(G_{uv})}
$
and let $\hat F_{uv}$ be the cut-set of $G_{uv}$ corresponding to cut $[\hat
S_{uv} ,\hat T_{uv}]$.

Observe that for distinct edges $\{u,v\}$ and $\{s,t\}$, it is possible that ${|\hat
F_{uv}| \ne |\hat F_{st}|}$.
We will consider bonds $F$ for which all induced cut-sets $\hat F_{uv}$ have the
same size.

\begin{definition}
Let $\ell$ be a positive integer.
A bond~$F$ of $H$ with cut $[S, T]$ is said to be \emph{\mbox{$\ell$-uniform}}
if, \emph{(i)} $F$~is~nice, and $(ii)$ for every $j$, ${0 \le j \le h-1}$, and
every edge $\{u,v\} \in E(H^j)$ with weight~$1$ such that ${|\{u,v\} \cap S\}| =
1}$, $|\hat F_{uv}| = \ell$.
\end{definition}

An $\ell$-uniform bond induces a cut-set of $G$ of size $\ell$.

\begin{lemma}\label{lem-induced-cut}
Suppose $F$ is an $\ell$-uniform bond of $H$. One can find in polynomial time a
cut-set $L$ of $G$ with $|L| = \ell$.
\end{lemma}

\begin{proof}
Let $u,v$ be the vertices of $K_2$ to which $\xig$ was applied. Since $F$ is
$\ell$-uniform, $|\hat F_{uv}|= \ell$. Note that $\hat F_{uv}$ induces a cut-set
of size $\ell$ on $G$.
\end{proof}

In the opposite direction, a cut of $G$ induces an $\ell$-uniform bond~of~$H$.

\begin{lemma}\label{lem-induced-bond}
Suppose $L$ is a cut-set of $G$ with $|L| = \ell$. One can find in polynomial
time an $\ell$-uniform bond $F$ of $H$ with $w(F) = \ell^h$.
\end{lemma}

\begin{proof}
For each $j \ge 0$, we construct a bond~$F^j$~of~$H^j$.
For $j = 0$, let $F^0$ be the set containing the unique edge of $H^0 = K_2$.
Suppose now that we already constructed a bond~$F^{j-1}$ of~$H^{j-1}$. For each
edge $\{u,v\} \in F^{j-1}$, let $L_{uv}$ be the set of edges of $G_{uv}$
corresponding to~$L$. Define $F^j := \cup_{\{u,v\} \in F^{j-1}} L_{uv}$.
One can verify that indeed $F^j$ is a bond of $H^j$, and that $w(F_j) = |L|
\times w(F_{j-1}) = \ell^j$.
\end{proof}

\begin{lemma} \label{lem-uniform}
There is a polynomial-time algorithm that receives a bond~$F$~of~$H$, and finds
an $\ell$-uniform bond ${F'}$~of~$H$ such that $w(F') = \ell^h \ge w(F)$.
\end{lemma}

\begin{proof}
Let $[S,T]$ be the cut corresponding to $F$. First, find the largest cut-set of
a graph~$G_{uv}$ over cut-sets~$\hat F_{uv}$. More precisely, define $\hat F$ to
be the cut-set $\hat F_{uv}$ with maximum $|\hat F_{uv}|$ over all edges $\{u,v\}
\in E(H^j)$ with weight~$1$ such that ${|\{u,v\} \cap S\}| = 1}$, and over all
$j$, ${0 \le j \le h-1}$.
Let $\ell := |\hat F|$.

We claim that for every $j$, $0 \le j \le h$, and every edge $\{u,v\} \in E(H^j)$
with weight~$1$ such that ${|\{u,v\} \cap S\}| = 1}$, $w(F_{uv}) \le \ell^{h -
j}$.
The proof is by (backward) induction on $j$. For $j = h$, $F_{uv} = \{u,v\}$, so
$w(F_{uv}) = 1$. Next, let $j < h$, and assume the claim holds for $j+1$.

Let $F^0_{uv}$ be the subset of edges in $F_{uv}$ incident with $u$ or $v$.
The set $F_{uv}$ can be partitioned into $F^0_{uv}$ and sets $F_{st}$ for
$\{s,t\} \in \hat F_{uv}$. To see this, observe that each edge $\{x, y\} \in
F_{uv} \setminus F^0_{uv}$ must be incident with descendants of $\{u,v\}$, and
thus $\{x, y\}$ is incident with vertices of~$V_{st}$, for some edge $\{s,t\}
\in E(G_{uv})$. Since $|\{x, y\}\cap S| = 1$, neither $V_{st} \cup \{s, t\}
\subseteq S$, nor $V_{st} \cup \{s, t\} \subseteq T$. Because $F$ is nice, it
follows that $|\{s,t\} \cap S| = 1$, then $\{s,t\} \in \hat F_{uv}$, and thus
$\{x,y\} \in F_{st}$.
To complete the claim, observe that, by the induction hypothesis, ${w(F_{st}) \le
\ell^{h - j - 1}}$ for each $\{s,t\} \in \hat F_{uv}$, and recall that $|\hat
F_{uv}| \le |\hat F|$. Therefore
\[
w(F) = w(F^0_{uv}) + \sum_{\{s,t\} \in \hat F_{uv}} w(F_{st})
     \le |\hat F| \times \ell^{h - j - 1} =  \ell^{h - j}.
\]

Using Lemma~\ref{lem-induced-bond} for $\hat F$, we construct a bond~$F'$
for~$H$ with ${w(F') = \ell^h}$.
\end{proof}

\begin{lemma}\label{lem-opt-sol}
Let $F^*$ be a bond of $H$ with maximum weight. Then $w(F^*) = k^h$.
\end{lemma}

\begin{proof}
We assume that $F^*$ is $\ell$-uniform such that $w(F^*) = \ell^h$ for some
$\ell$; if this is not the case, then use Lemma~\ref{lem-uniform}.

Since $F^*$ is $\ell$-uniform, using Lemma~\ref{lem-induced-bond} one obtains a
cut-set $L$ of $G$ with size~$\ell$, then~$\ell \le k$, and thus $w(F^*) \le
k^h$.

Conversely, let $L$ be a cut-set of $G$ with size $k$. Using
Lemma~\ref{lem-induced-bond} for $L$, we obtain a bond~$F$ of~$H$ with weight
$k^h$, and thus $w(F^*) \ge k^h$.
\end{proof}

\begin{lemma}\label{lem-weighted-hard}
If there exists a constant-factor approximation algorithm for {\sc Weighted
Largest Bond}, then $\PP = \NP$.
\end{lemma}

\begin{proof}
Consider a graph $G$ whose maximum cut has size~$k$. Construct graph~$H$ and
obtain a bond $F$~of~$H$ using an $\alpha$-approximation, for some constant ${0
< \alpha < 1}$. \mbox{Using} the algorithm of Lemma~\ref{lem-uniform}, obtain an
$\ell$-uniform bond $F'$~of~$H$ such that $w(F') = \ell^h \ge w(F)$.
Using Lemma~\ref{lem-opt-sol} and the fact that $F'$ is an
$\alpha$-approximation, ${\ell^h \ge \alpha \times k^h}$.
Using Lemma~\ref{lem-induced-cut}, one can obtain a cut-set $L$ of $G$ with size
$\ell \ge \alpha^{\frac{1}{h}} k$.

For any constant $\varepsilon$, ${0 < \varepsilon < 1}$, we can set $h =
\lceil\log_{1 - \varepsilon} \alpha \rceil$, such that the cut-set~$L$ has size
at least $\ell \ge (1 - \varepsilon) k$. Since {\sc Maximum Cut} is APX-hard,
this implies $\PP = \NP$.
\end{proof}

\begin{theorem}
If there exists a constant-factor approximation algorithm for {\sc Largest
Bond}, then $\PP = \NP$.
\end{theorem}

\begin{proof}
We show that if there exists an $\alpha$-approximation algorithm for {\sc
Largest Bond}, for constant $0 < \alpha < 1$, then there is an
$\alpha/2$-approximation algorithm for the {\sc Binary Weighted Largest Bond}, so the
theorem will follow from Lemma~\ref{lem-weighted-hard}.

Let $H$ be a weighted graph whose edge weights are all $0$~or~$1$. Let $m$ be
the number of edges with weight~$0$, and let $l$ be the weight of a bond of
$H$ with maximum weight.
Assume $l\ge 2/\alpha$, as otherwise, one can find an optimal solution in
polynomial time by enumerating sets of up to $2/\alpha$ edges.

Construct an unweighted graph $G$ as follows. Start with a copy of $H$ and, for
each edge $\{u,v\} \in E(H)$ with weight $1$, replace $\{u,v\} \in E(G)$ by $m$
parallel edges. Finally, to obtain a simple graph, subdivide each edge of $G$.
If $F$ is a bond of~$G$, then one can construct a bond~$F'$ of~$H$ by undoing
the subdivision and removing the parallel edges. Each edge of $F'$ has weight
$1$, with exception of at most $m$ edges. Thus,
$
w(F') \ge (|F| - m) / m.
$

Observe that an optimal bond of $H$ induces a bond of $G$ with size at
least~$ml$. Thus, if~$F$ is an $\alpha$-approximation for $G$, then ${|F| \ge
\alpha m l}$ and therefore
\[
w(F') \ge \frac{\alpha m l - m}{m} = \alpha l - 1 \ge \alpha l - \alpha l / 2 = \alpha l / 2.
\]
We conclude that $F'$ is an $\alpha/2$-approximation for $H$.
\end{proof}

\section{Algorithmic upper bounds for clique-width parameterization}

Lemma~\ref{lb_cw} shows that {\sc Largest Bond} on graphs of clique-width $w$ cannot be solved in time $f(w)\times n^{o(w)}$ unless the ETH fails. Now, we show that given an expression tree of width at most $w$, {\sc Largest Bond} can be solved in $f(w)\times n^{O(w)}$ time.


An expression tree $\mathcal{T}$ is irredundant if for any join node $\eta(i,j)$, the vertices labeled
by~$i$ and~$j$ are not adjacent in the graph associated with its child. It was shown by Courcelle and Olariu~\cite{courcelle2000upper} that every expression tree $\mathcal{T}$ of $G$ can be transformed into an irredundant expression tree $\mathcal{T}$ of the same width in time linear in the size of~$\mathcal{T}$. Therefore, without loss of generality, we can assume that $\mathcal{T}$ is irredundant.

Our algorithm is based on dynamic programming over the expression tree of the input graph. We first describe what we store in the tables corresponding to the nodes in the expression tree.

Given a $w$-labeled graph $G$, two connected components of $G$ has the same {\em type} if they have the same set of labels. Thus, a $w$-labeled graph $G$ has at most $2^w-1$ types of connected components.

Now, for every node $X_\ell$ of $\mathcal{T}$, denote by $G_{X_\ell}$ the $w$-labeled graph associated with this node, and let $L_1(X_\ell),\ldots,L_w(X_\ell)$ be the sets of vertices of $G_{X_\ell}$ labeled with $1,\ldots,w$, respectively. We define a table where each entry is of the form $c[\ell,s_{1}, ..., s_{w}, r, e_{1}, ..., e_{2^w-1}, d_{1}, ..., d_{2^w-1}]$, such that: $0\leq s_i \leq |L_i(X_\ell)|$ for $1\leq i\leq w$; $0\leq r\leq |E(G_{X_\ell})|$; $0\leq e_{i}\leq \min\{2,|L_i(X_\ell)|\}$ for $1\leq i\leq 2^w-1$; and $0\leq d_{i}\leq \min\{2,|L_i(X_\ell)|\}$ for $1\leq i\leq 2^w-1$.

Each entry of the table represents whether there is a partition $V_1,V_2$ of $V(G_{X_\ell})$ such that: $|V_1\cap L_i(G_{X_\ell})| = s_i$; the cut-set of $[V_1,V_2]$ has size at least $r$;  $G_{X_\ell}[V_1]$ has $e_i$ connected components of type $i$; $G_{X_\ell}[V_2]$ has $d_i$ connected components of type $i$, where $e_i=2$ means that $G_{X_\ell}[V_1]$ has {\em at least} two connected components of type $i$. The same holds for $d_i$.

Notice that this table contains $f(w)\times n^{\Oh(w)}$ entries.
If $X_\ell$ is the root node of $\mathcal{T}$ (that is, $G = G_{X_\ell})$, then the size of the largest bond of $G$ is equal to the maximum value of $r$ for which the table for $X_\ell$ contains a valid entry (true value), such that there are $j$ and $k$ such that $e_{i} = 0$, $e_{j} = 1$ for $1 \leq i, j \leq 2^w-1$, $i \neq j$; and $d_{i} = 0$, $d_{k} = 1$ for $1\leq i, k \leq 2^w-1$,~$i \neq k$.

It is easy to see that we store enough information to compute a largest bond. Note that a $w$-labeled graph is connected if and only if it has exactly one type of connected components and exactly one component of such a type.

Now we provide the details of how to construct and update such tables. The construction for introduce nodes of $\mathcal{T}$ is straightforward.

\textbf{Relabel node:} Suppose that $X_\ell$ is a relabel node $\rho(i,j)$, and let $X_{\ell'}$ be the child of
$X_\ell$. Then the table for $X_\ell$ contains a valid entry $c[\ell,s_{1}, ..., s_{w}, r, e_{1}, ..., e_{2^w-1}, d_{1}, ..., d_{2^w-1}]$ if and only if  the table for $X_{\ell'}$ contains an entry $c[\ell',s'_{1}, ..., s'_{w}, r, e'_{1}, ..., e'_{2^w-1}, d'_{1}, ..., d'_{2^w-1}]=true$, where:
$s_i = 0$;
$s_{j}=s'_{i}+s'_{j}$;
$s_{p} = s'_{p}$ for $1\leq p \leq w$, $p\neq i, j$;
$e_p=e'_p$ for any type that contain neither $i$ nor $j$;
$e_p=0$ for any type that contains $i$;
and for any type $e_p$ that contains $j$, it holds that $e_p=\min\{2, e'_p+e'_q+e'_r\}$ where
$e'_q$ represent the set of labels $(C_p\setminus \{j\})\cup \{i\}$,
$e'_r$ represent the set of labels $C_p\cup \{i\}$, and
$C_p$ is the set of labels associated to $p$.
The same holds for $d_{1}, ..., d_{2^w-1}$.

\textbf{Union node:} Suppose that $X_\ell$ is a union node with children $X_{\ell'}$ and $X_{\ell''}$.
It holds that $c[\ell,s_{1}, ..., s_{w}, r, e_{1}, ..., e_{2^w-1}, d_{1}, ..., d_{2^w-1}]$ equals true if and only if there are valid entries $c[\ell', s'_{1}, ..., s'_{w}, r', e'_{1}, ..., e'_{2^w-1}, d'_{1}, ..., d'_{2^w-1}]$ and $c[\ell'',s''_{1}, ..., s''_{w}, r'', e''_{1}, ..., e''_{2^w-1}, d''_{1}, ..., d''_{2^w-1}]$, having:
$s_i=s'_{i}+s''_{i}$ for $1 \leq i \leq w$;
$r' + r'' \geq r$;
$e_{k} = \min\{2, e'_{k}+e''_{k}\}$, and
$d_{k} = min\{2, d'_{k} + d''_{k}\}$ for $1\leq k \leq 2^w-1$.

\textbf{Join node:}
Finally, let $X_\ell$ be a join node $\eta(i,j)$ with the child $X_{\ell'}$. Remind that since the expression tree is irredundant then the vertices labeled by $i$ and $j$ are not adjacent in the graph $G_{X_{\ell'}}$. Therefore, the entry $c[\ell,s_{1}, ..., s_{w}, r, e_{1}, ..., e_{2^w-1}, d_{1}, ..., d_{2^w-1}]$ equals true if and only if there is a valid entry $c[\ell',s_{1}, ..., s_{w}, r', e'_{1}, ..., e'_{2^w-1}, d'_{1}, ..., d'_{2^w-1}]$ where
$$r' + s_{i}\times(|L_{j}(X_{\ell'})| - s_{j}) + s_{j}\times(|L_{i}(X_{\ell'})| - s_{i}) \geq r,$$
and
$e_p=e'_p$, case $p$ is associated to a type that contains neither $i$ nor $j$;
$e_{p} = 1$, case $p$ is associated to $C^{\ell'}_{i,j}\setminus \{i\}$, where $C^{\ell'}_{i,j}$ is the set of labels obtained by the union of the types of $G_{X_{\ell'}}$ with some connected component having either label $i$ or label $j$;
$e_{p} = 0$, otherwise.
The same holds for $d_{1}, ..., d_{2^w-1}$.

The correctness of the algorithm follows from the description of the procedure. Since for each $\ell$, there are $\Oh((n+1)^w\times m\times (3^{2^w-1})^2)$ entries, the running time of the algorithm is $f(w)\times n^{\Oh(w)}$. This algorithm  together with Lemma~\ref{lb_cw} concludes the proof of the Theorem~\ref{cliquewidth_up}.

\begin{theorem}\label{cliquewidth_up}
{\sc Largest Bond} cannot be solved in time $f(w)\times n^{o(w)}$ unless ETH fails, where $w$ is the clique-width of the input graph. Moreover, given an expression tree of width at most $w$, {\sc Largest Bond} can be solved in time $f(w)\times n^{\Oh(w)}$.
\end{theorem}

In order to extend this result to {\sc Largest $st$-Bond}, it is enough to observe that given a tree expression $\mathcal{T}$ of $G$ with width $w$, it is easy to construct a tree expression $\mathcal{T}'$ with width equals $w+2$, where no vertex of $V(G)$ has the same label than either $s$ or $t$. Let $w+1$ be the label of $s$, and let $w+2$ be the label of $t$. By fixing, for each $\ell$, $s_{w+1}=|L_{w+1}(X_\ell)|$ and~$s_{w+2}=0$, one can solve {\sc Largest $st$-Bond} in time $f(w)\times n^{\Oh(w)}$.

\section{Bounding the treewidth of G}

In the remainder of this paper we deal with our main problems: {\sc Largest Bond} and {\sc Largest $st$-Bond} parameterized by the size of the solution ($k$).
Inspired by the principle of preprocessing the input to obtain a kernel, we consider the strategy of preprocessing the input in order to bound the treewidth of the resulting instance.

We start our analysis with {\sc Largest Bond}.

\begin{definition}
A graph $H$ is called a minor of a graph $G$ if $H$ can be formed from $G$ by deleting edges, deleting vertices, and by contracting edges. For each vertex $v$ of $H$, the set of vertices of $G$ that are contracted into $v$ is called a branch set of $H$.
\end{definition}

\begin{lemma}\label{k2kimpliesbond}
Let $G$ be a simple connected undirected graph, and $k$ be a positive integer. If $G$ contains $K_{2,k}$ as a minor then $G$ has a bond of size at least $k$.
\end{lemma}
\begin{proof}
Let $H$ be a minor of $G$ isomorphic to $K_{2,k}$. Since $G$ is connected and each branch set of $H$ induces a connected subgraph of $G$, from $H$ it is easy to construct a bond of $G$ of size at least $k$.
\end{proof}

Combined with Lemma~\ref{k2kimpliesbond}, the following results show that, without loss of generality, our study on $k$-bonds can be reduced to graphs of treewidth $\Oh(k)$.

\begin{lemma}\cite{bodlaender1997interval}
Every graph $G=(V, E)$ contains $K_{2,k}$ as a minor or has treewidth at most $2k-2$.
\end{lemma}

\begin{lemma}\cite{bodlaender1997interval}\label{algoK2k}
There is a polynomial-time algorithm that either concludes that the input graph $G$ contains $K_{2,k}$ as a minor, or outputs a tree-decomposition of $G$ of width at most $2k-2$.
\end{lemma}

From Lemma~\ref{k2kimpliesbond} and Lemma~\ref{algoK2k} it follows that there is a polynomial-time algorithm that either concludes that the input graph $G$ has a bond of size at least $k$, or outputs a tree-decomposition of $G$ of width at most $2k-2$.

\subsection{The st-bond case}

Let $S\subseteq V(G)$ and let $\partial(S)$ be a bond of a connected graph $G$. Recall
that a block is a $2$-vertex-connected subgraph of $G$ which is inclusion-wise
maximal, and a block-cut tree of $G$ is a tree whose vertices represent the blocks and the cut vertices of $G$, and there is an edge in the block-cut tree for each pair of a block and a cut vertex that belongs to that block.
Then, $\partial(S)$ intersects at most one block of~$G$. More
precisely, for any two distinct blocks $B_1$ and $B_2$ of~$G$, if $S \cap V(B_1)
\neq \emptyset$ and $S \cap V(B_1) \neq V(B_1)$, then either $V(B_2) \subseteq
S$, or $V(B_2) \subseteq V \setminus S$. Indeed, if this is not the case, then
either $G[S]$ or $G[V\setminus S]$ would be disconnected.
Thus, to solve {\sc Largest $st$-Bond}, it is enough to consider, individually, each block on the path between $s$ and $t$ in the block-cut tree of $G$. Also, if a block is composed of a single edge, then it is a bridge~in~$G$, which is not a solution for the problem unless $k = 1$.
Thus, we may assume without loss of generality that $G$ is
$2$-vertex-connected.

\begin{lemma} \label{twopaths}
Let $G$ be a $2$-vertex-connected graph. For all $v \in V(G)\setminus\{s,t\}$, there is an $sv$-path and a $tv$-path which are internally disjoint.
\end{lemma}

\begin{proof}
Since $G$ is $2$-vertex-connected, there are two disjoint $sv$-paths $P_s$ and $P'_s$ and there is a $tv$-path $P_t'$ which does not include $s$. Let $x$ be the first vertex of $P_t'$ which belongs to $V(P_s \cup P_s')$ and assume, w.l.o.g., that $x \in P_s'$. Let $P_t''$ be the sub-path of $P_t'$ from $t$ to $x$ and $P_s''$ the sub-path of $P_s'$ from $x$ to $v$. Now define $P_t$ as $tP_t''xP_s''v$ and notice that $P_t$ is a $tv$-path disjoint from $P_s$.
\end{proof}


\begin{lemma}\label{st-minor}
Let $G$ be a $2$-vertex-connected graph. If $G$ contains $K_{2,2k}$ as a minor,
then there exists $S \subseteq V(G)$ such that $\partial(S)$ is a bond of size
at least $k$.
\end{lemma}

\begin{proof}
Let $G$ be a graph containing a $K_{2,2k}$ as a minor. If $k = 1$, the
statement holds trivially, thus assume $k \ge 2$. Also, since $G$ is connected,
one can assume that this minor was obtained by contracting or removing edges
only, and thus its branch sets contain all vertices of~$G$.
Let $A$ and $B$ be the branch sets corresponding to first side of $K_{2,2k}$,
and let $X_1, X_2, \dots, X_{2k}$ be the remaining branch sets.

First, suppose that $s$ and $t$ are in distinct branch sets. If this is the
case, then there exist distinct indices $a, b \in \{1, \dots, 2k\}$ such that $s
\in A \cup X_a$ and $t \in B \cup X_b$. Now observe that $G[A \cup X_a]$ and
$G[B \cup X_b]$ are connected, which implies an $st$-bond with at least $2k - 1 \ge
k$ edges. Now, suppose that $s$ and $t$ are in the same branch set. In this
case, one can assume without loss of generality that $s, t \in A \cup X_{2k}$.

Define $U = A \cup X_{2k}$ and $Q = V(G) \setminus U$. Observe that $G[U]$ and
$G[Q]$ are connected. Consider an arbitrary vertex $v$ in the set $Q$. Since $G$
is $2$-vertex-connected, Lemma~\ref{twopaths} implies that there exist an
$sv$-path $P_s$ and a $tv$-path $P_t$ which are internally disjoint. Let $P'_s$
and $P'_t$ be maximal prefixes of $P_s$ and $P_t$, respectively, whose vertices
are contained in $U$.

We partition the set $U$ into parts $U_s$ and $U_t$ such that $G[U_s]$ and
$G[U_t]$ are connected.
Since $G[U]$ is connected, there exists a tree $T$ spanning $U$.
Direct all edges of $T$ towards $s$ and
partition $U$ as follows.
Every vertex in $P'_s$ belongs to $U_s$ and every vertex in $P'_t$  belongs to
$U_t$. For a vertex $u \notin V(P'_s \cup P'_t)$, let $w$ be the first
ancestor of $u$ (accordingly to $T$) which is in $P'_s \cup P'_t$. Notice that
$w$ is well-defined since $u \in V(T)$ and the root of $T$ is $s \in V(P'_s \cup
P'_t)$.
Then $u$ belongs to $U_s$ if $w \in V(P'_s)$, and $u$ belongs to $U_t$ if $w \in
V(P'_t)$.

Observe that that there are at least $2k-1$ edges between $U$ and $Q$, and thus
there are at least $k$ edges between $U_s$ and $Q$, or between $U_t$ and $Q$.
Assume the former holds, as the other case is analogous.
It follows that $G[U_s]$ and $G[U_t \cup Q]$ are connected and induce a bond
of~$G$ with at least $k$ edges.
\end{proof}

Lemma~\ref{algoK2k} and Lemma~\ref{st-minor} imply that there is an algorithm
that either concludes that the input graph $G$ has a bond of size at least $k$,
or outputs a tree-decomposition of an equivalent instance $G'$ of width
$\Oh(k)$.

\begin{corollary}\label{st_bounding_treewidth}
Given a graph $G$, vertices $s,t\in V(G)$, and an integer $k$, there exists a polynomial-time algorithm that either concludes that $G$ has an $st$-bond of size at least $k$ or outputs a subgraph $G'$ of $G$ together with a tree decomposition of $G'$ of width equals $\Oh(k)$, such that $G'$ has an $st$-bond of size at least $k$ if and only if $G$ has an $st$-bond of size at least~$k$.
\end{corollary}

\begin{proof}
Find a block-cut tree of $G$ in linear time~\cite{CormenLRS01}, and let $B_s$
and $B_t$ be the blocks of $G$ that contain $s$ and $t$, respectively. Remove
each block that is not in the path from $B_s$ to $B_t$ in the block-cut tree
of~$G$. Let $G'$ be the remaining graph.
For each block $B$ of $G'$, consider the vertices $s'$ and $t'$ of $B$ which are
nearest to $s$ and $t$, respectively. Using
Lemmas~\ref{algoK2k}~and~\ref{st-minor} one can in polynomial time either
conclude that $B$ has an \mbox{$s't'$-bond}, in which case $G$ is a
yes-instance, or compute a tree decomposition of $B$ with width at most
$\Oh(k)$.

Now, construct a tree decomposition of $G'$ as follows. Start with the union of
the tree decompositions of all blocks of $G'$. Next, create a bag $\{u\}$ for
each cut vertex $u$ of $G'$.
Finally, for each cut vertex $u$ and any bag corresponding to a block $B$
connected through $u$, add an edge between~$\{u\}$ and one bag of the tree
decomposition of $B$ containing $u$.
Note that this defines a tree decomposition of $G'$ and that each bag has at
most $\Oh(k)$ vertices.
\end{proof}


Note that since $k$-bonds are solutions for {\sc Connected Max Cut}, the results presented in this section naturally apply to such a problem as well.

\section{Taking the treewidth as parameter}

In the following, given a tree decomposition $\mathcal{T}$, we denote by $\ell$ one node of~$\mathcal{T}$ and by $X_\ell$ the vertices contained in the \emph{bag} of $\ell$. We assume w.l.o.g that $\mathcal{T}$ is a extended version of a \emph{nice} tree decomposition (see~\cite{cygan2015parameterized}), that is, we assume that there is a special root node~$r$ such that $X_\ell = \emptyset$ and all edges of the tree are directed towards $r$ and each node~$\ell$ has one of the following five types: {\em Leaf}\,; {\em Introduce vertex}; {\em Introduce edge}; {\em Forget vertex}; and {\em Join}.
%
%
%
%
%
%
Moreover, define $G_\ell$ to be the subgraph of $G$ which contains only vertices and edges that have been introduced in $\ell$ or in a descendant of $\ell$.

The number of partitions of a set of $k$ elements is the \textit{$k$-th Bell number}, which we denote by $B(k)$ ($B(k)\leq k!$~\cite{bell}).

\begin{theorem}\label{dynamic}
Given a nice tree decomposition of~$G$ with width $tw$, one can find a bond of maximum size in time
$2^{\Oh(tw\log{tw})}\times n$
where $n$ is the
number of vertices of~$G$.
\end{theorem}

\begin{proof}
Let $\partial_{G}(U)$ be a bond of $G$, and $[U, V\setminus U]$ be the cut defined by such a bond.
Set $S^\ell_U=U\cap X_\ell$. The removal of $\partial_{G}(U)$ partitions $G_\ell[U]$ into a set $C^\ell_{U}$ of connected components, and $G_\ell[V\setminus U]$ into a set $C^\ell_{V\setminus U}$ of connected components. Note that~$C^\ell_{U}$ and~$C^\ell_{V\setminus U}$ define partitions of $S^\ell_U$ and $X_\ell\setminus S^\ell_U$, denoted by $\rho^\ell_1$~and $\rho^\ell_2$ respectively, where the intersection of each connected component of $C^\ell_{U}$ with $S^\ell_U$ corresponds to one part of $\rho^\ell_1$. The same holds for~$C^\ell_{V\setminus U}$ with respect to $X_\ell \setminus S^\ell_U$ and $\rho^\ell_2$.

We define a table for which an entry $c[\ell, S, \rho_1, \rho_2]$ is the size of a largest cut-set (partial solution) of the subgraph~$G_\ell$, where~$S$ is the subset of $X_\ell$ to the left part of the bond, $ X_\ell \setminus S$ is the subset to the right part, and
$\rho_1$, $\rho_2$ are the partitions of $S$ and  $X_\ell \setminus S$ representing, after the removal of the partial solution, the intersection with the connected components to the left and to the right, respectively.
If there is no such a partial solution then $c[\ell, S, \rho_1, \rho_2] = -\infty$.

For the case that $S$ is empty, two special cases may occur: either $U\cap V(G_\ell) = \emptyset$,  in which case there are no connected components in $C^\ell_U$, and thus ${\rho_1 = \emptyset}$; or $C^\ell_U$ has only one connected component which does not intersect $X_\ell$, i.e., $\rho_1=\{\emptyset\}$, this case means that the connected component in $C^\ell_U$ was completely forgotten.
Analogously, we may have $\rho_2 = \emptyset$ and $\rho_2=\{\emptyset\}$.
%
%
Note that we do not need to consider the case $\{\emptyset\}\subsetneq \rho_i$ since it would imply in a disconnected solution. The largest bond of a connected graph $G$ corresponds to the root entry $c[r, \emptyset, \{\emptyset\}, \{\emptyset\}]$.

%
To describe a dynamic programming algorithm, we only need to present the recurrence relation for each node type.


\medskip

\textbf{Leaf:} In this case, $X_\ell = \emptyset$. There are a few combinations for $\rho_1$ and $\rho_2$: either $\rho_1 = \emptyset$, or $\rho_1 = \{\emptyset\}$, and either $\rho_2 = \emptyset$, or $\rho_2 = \{\emptyset\}$.
Since for this case  $G_\ell$ is empty, there can be no connected components, so having $\rho_1 = \emptyset$ and $\rho_2 = \emptyset$ is the only feasible choice.
\[
c[\ell, S, \rho_1, \rho_2] =
\begin{cases}
0  & \text{if $S = \emptyset$, $\rho_1 = \emptyset$ and $\rho_2 = \emptyset$},\\
-\infty & \text{otherwise}.
\end{cases}
\]

\smallskip

\textbf{Introduce vertex:} 
We have only two possibilities in this case, either $v$ is an isolated vertex to the left ($v \in S$) or it is an isolated vertex to the right ($v \notin S$). Thus, a partial solution on $\ell$ induces a partial solution on $\ell'$, excluding $v$ from its part. 
\[
c[\ell, S, \rho_1, \rho_2] =
\begin{cases}
c[\ell', S\setminus\{v\}, \rho_1\setminus\{\{v\}\}, \rho_2] & \text{if $\{v\} \in \rho_1$}, \\
c[\ell', S, \rho_1, \rho_2\setminus\{\{v\}\}] & \text{if $\{v\} \in \rho_2$}, \\
- \infty & \text{if $\{v\} \notin \rho_1 \cup \rho_2$}. \\
\end{cases}
\]

\textbf{Introduce edge:} In this case, either the edge $\{u,v\}$ that is being inserted is incident with one vertex of each side, or the two endpoints are at the same side. In the former case, a solution on $\ell$ corresponds to a solution on $\ell'$ with the same partitions, but with value increased. In the latter case, edge~$\{u,v\}$ may connect two connected components of a partial solution on $\ell'$.
\[
c[\ell, S, \rho_1, \rho_2] =
\begin{cases}
c[\ell', S, \rho_1, \rho_2] + 1 & \text{if $u \in S$ and $v \notin S$ or $u \notin S$ and $v \in S$} ,\\
max_{ \rho_1'}\{c[\ell', S, \rho_1', \rho_2]\} & \text{if $u \in S$ and $v \in S$}, \\
max_{ \rho_2'}\{c[\ell', S, \rho_1, \rho_2']\} & \text{if $u \notin S$ and $v \notin S$}.
\end{cases}
\]
Here, $\rho_1'$ spans over all refinements of $\rho_1$ such that the union of the parts containing $u$ and $v$ results in the partition $\rho_1$. The same holds for $\rho_2'$.


\smallskip

\textbf{Forget vertex:} In this case, either the forgotten vertex $v$ is in the left side of the partial solution induced on $\ell$, or is in the right side. Thus, $v$ must be in the connected component which contains some part of $\rho_1$, or some part of $\rho_2$. We select the possibility that maximizes the value
\[
c[\ell, S, \rho_1, \rho_2] = max_{ \rho_1',\rho_2'}\{ c[\ell', S \cup \{v\}, \rho_1', \rho_2], c[\ell', S, \rho_1, \rho_2'] \}.
\]
Here, $\rho_1'$ spans over all partitions obtained from~$\rho_1$ by adding $v$ in some part of~$\rho_1$ (if $\rho_1=\{\emptyset\}$ then $\rho_1'=\{v\}$). The same holds for $\rho_2'$.


\smallskip

\textbf{Join:} This node represents the join of two subgraphs $G_{\ell'}$ and $G_{\ell''}$ and ${X_\ell = X_{\ell'} = X_{\ell''}}$.
By counting the bond edges contained in $G_{\ell'}$ and in
$G_{\ell''}$, each edge is counted at least once, but edges in $X_{\ell}$ are
counted twice. Thus
\[
c[\ell, S, \rho_1, \rho_2] = max\{c[\ell', S, \rho_1', \rho_2'] + c[\ell', S, \rho_1'', \rho_2'']\} - |\{\{u,v\} \in E, u \in S, v \in X_{\ell} \setminus S\}|.
\]
In this case, we must find the best combination between the two children. Namely, for $i\in\{1,2\}$, we consider combinations of $\rho_i'$ with $\rho_i''$ which merge into $\rho_i$. If $\rho_i=\{\emptyset\}$ then either $\rho_i'=\{\emptyset\}$ and $\rho_i''=\emptyset$; or $\rho_i'=\emptyset$ and $\rho_i''=\{\emptyset\}$. Also, if $\rho_i=\emptyset$ then $\rho_i'=\emptyset$ and $\rho_i''=\emptyset$.

\bigskip

The running time of the dynamic programming algorithm can be estimated as follows. The number of nodes in the decomposition is $\Oh(tw\times n)$~\cite{cygan2015parameterized}.
For each node~$\ell$, the parameters $\rho_1$ and $\rho_2$ induce a partition of $X_\ell$; the number of partitions of $X_\ell$ is given by the corresponding Bell number, $B(|X_\ell|) \le B(tw+1)$.
Each such a partition $\rho$ corresponds to a number of choice of parameter $S$ that corresponds to a subset of the parts of $\rho$; thus the number of choices for $S$ is not larger than ${2^{|\rho|} \le 2^{|X_\ell|} \le 2^{tw+1}}$.
Therefore, we conclude that the table size is at most $\Oh(B(tw+1) \times 2^{tw} \times tw \times n)$.
Since each entry can be computed in $2^{\Oh(tw\log tw)}$ time, the total complexity is $2^{\Oh(tw\log{tw})}\times n$. The correctness of the recursive formulas is straightforward.
\end{proof}

The reason for the $2^{\Oh(tw \log tw)}$ dependence on treewidth is because we
enumerate all partitions of a bag to check connectivity.
However, one can obtain single exponential-time dependence by modifying the
presented algorithm using techniques based on Gauss elimination, as described
in~\cite[Chapter~11]{cygan2015parameterized} for {\sc Steiner Tree}.

\begin{theorem}\label{th:st_bond_tw}
{\sc Largest $st$-Bond} is fixed-parameter tractable when parameterized by treewidth.
\end{theorem}
\begin{proof}
The solution of {\sc Largest $st$-Bond} can be found by a dynamic programming as presented in Theorem~\ref{dynamic} where we add $s$ and $t$ in all the nodes and we fix~$s\in S$ and~$t\notin S$.
\end{proof}

\vspace*{-1mm}

Finally, the following holds.

\begin{corollary}
{\sc Largest Bond} and {\sc Largest $st$-Bond} are fixed-parameter tractable when parameterized by the size of the solution, $k$.
\end{corollary}
\begin{proof}
Follows from Lemma~\ref{k2kimpliesbond}, Lemma~\ref{algoK2k}, Corollary~\ref{st_bounding_treewidth}, Theorem~\ref{dynamic} and Theorem~\ref{th:st_bond_tw}.
\end{proof}

\section{Infeasibility of polynomial kernels}

As seen previously, any bond $\partial(S)$ of a graph $G$ intersects at most one of its block.
Thus, an or-composition for {\sc Largest Bond} parameterized by $k$ can be done from the disjoint union of~$\ell$ inputs, by selecting exactly one vertex of each input graph and contracting them into a single vertex.
Now, let $(G_1,k,s_1,t_1),(G_2,k,s_2,t_2),\ldots,(G_\ell,k,s_\ell,t_\ell)$ be $\ell$ instances of {\sc Largest $st$-Bond} parameterized by $k$.
An or-composition for {\sc Largest $st$-Bond} parameterized by $k$ can be done from the disjoint union of $G_1,G_2,\ldots,G_\ell$, by contracting $t_i,s_{i+1}$ into a single vertex, $1\leq i\leq \ell-1$, and setting $s=s_1$ and $t=t_\ell$.

Therefore, the following holds.

\begin{theorem}
{\sc Largest Bond} and {\sc Largest $st$-Bond} do not admit polynomial kernel unless NP $\subseteq$ coNP/poly.
\end{theorem}



\bibliography{bond}

\end{document}